\documentclass[onecollarge]{fbsstyle}
\smartqed  
\usepackage{graphicx}
\usepackage{mathptmx}      
\usepackage{amssymb}
\usepackage{amsmath}
\usepackage{bm}
\usepackage{mathptmx}
\usepackage{amssymb}
\usepackage{epsfig}
\newcommand{\bbC}{{\mathbb C}}
\newcommand{\bbR}{{\mathbb R}}

\newcommand{\bk}{\bm{k}}
\newcommand{\br}{\bm{r}}

\newcommand{\cB}{{\mathcal B}}

\newcommand{\fH}{\mathfrak{H}}
\newcommand{\fh}{\mathfrak{h}}

\newcommand{\Img}{\mathop{\rm Im}}
\newcommand{\lal}{{\langle}}
\newcommand{\ral}{{\rangle}}

\newcommand{\ri}{{\rm i}}

\DeclareSymbolFont{SY}{U}{psy}{m}{n}
\DeclareMathSymbol{\emptyset}{\mathord}{SY}{'306}
\numberwithin{equation}{section}
\begin{document}

\title{Unphysical energy sheets and resonances in the
Friedrichs-Faddeev model\thanks{This work was supported by the
Russian Foundation for Basic Research (grant 16-01-00706).}
}

\titlerunning{Unphysical energy sheets and resonances in the
Friedrichs-Faddeev model}        

\author{Alexander K. Motovilov}

\authorrunning{A.K.Motovilov} 

\institute{A.K.Motovilov \at
              Bogoliubov Laboratory of Theoretical Physics, JINR, Joliot-Curie 6,
              141980 Dubna, Russia  \\
             \emph{and} \\
              Dubna State University, Universitetskaya 19, 141980 Dubna, Russia
              \\
              \email{motovilv@theor.jinr.ru}}

\date{}

\maketitle

\begin{abstract}
We consider the Friedrichs-Faddeev model in the case where the
kernel of the potential operator is holomorphic in both arguments on
a certain domain of $\mathbb{C}$. For this model we, first, study
the structure of the $T$- and $S$-matrices on unphysical energy
sheet(s). To this end, we derive representations that explicitly
express them in terms of these same operators considered exclusively
on the physical sheet. Furthermore, we allow the Friedrichs-Faddeev
Hamiltonian undergo a complex deformation (or even a complex
scaling/rotation if the model is associated with an infinite
interval). Isolated non-real eigenvalues of the deformed Hamiltonian
are called the deformation resonances. For a class of
perturbation potentials with analytic kernels, we prove that the
deformation resonances do correspond to the scattering matrix
resonances, that is, they represent the poles of the scattering matrix
analytically continued to the respective unphysical energy sheet.
\end{abstract}

\vspace*{1cm}

\begin{center}
\textit{Dedicated to the memory of Ludwig Dmitrievich Faddeev}
\end{center}

\section{Introduction}


In 1938, Kurt Friedrichs \cite{Fr1938} considered a model
Hamiltonian of the form
$$
H_\epsilon=H_0+\epsilon V
$$
with $H_0$, the multiplication by the independent variable $\lambda$,
{$$
(H_0 f)(\lambda) =\lambda f(\lambda), \quad\quad \lambda\in(-1,1)\subset\bbR, \quad\quad f\in L_2(-1,1),
$$}
$\epsilon>0$, and $V$, an integral operator,
{$$
(Vf)(\lambda)=\int_{-1}^1 V(\lambda,\mu)f(\mu) d\mu,
$$}%
where the kernel $V(\lambda,\mu)$ is a continuous function in
$\lambda,\mu\in[a,b]$ of a H\"older class. Furthermore, he assumed
that $V(-1,\mu)=V(1,\mu)=V(\lambda,-1)=V(\lambda,1)=0$\, for any
$\lambda,\mu\in[-1,1]$.

The Hermitian (self-adjoint) operator $H_0$ has absolutely
continuous spectrum that fills the segment $[-1,1]$. Friedrichs
studied what happens to the continuous spectrum of $H_0$ under the
perturbation $\epsilon V$. He has proven that if $\epsilon$ is
sufficiently small then $H_\epsilon$ and $H_0$ are similar, which
means that the spectrum of $H_\epsilon$ is also absolutely
continuous and fills $[-1,1]$.

In a 1948 paper \cite{Fr1948},  Friedrichs has extended this result
to the case where the unperturbed Hamiltonian $H_0$ is the
multiplication by independent variable in the Hilbert space
$\fH=L_2(\Delta,\fh)$ of square-integrable vector-valued functions
$f:\,\Delta\to\fh$ where $\Delta$ is a finite or infinite interval
on the real axis, $\Delta=(a,b)$, with $-\infty\leq a < b \leq
+\infty$, and $\fh$ is an auxiliary Hilbert space (finite- or
infinite-dimensional). In this case, it is assumed that for every
{$\lambda,\mu\in\Delta$} the quantity {$V(\lambda,\mu)$} {is a
bounded linear operator on} {$\fh$}, that
$V(\lambda,\mu)=V(\mu,\lambda)^*$, and that $V$ is a H\"older
continuous operator-valued function of $\lambda,\mu$. Friedrichs
proved that for sufficiently small $\epsilon$ the perturbed operator
$H_\epsilon=H_0+\epsilon V$ is unitarily equivalent to the
unperturbed one, $H_0$, and thus the spectrum of $H_\epsilon$ is
absolutely continuous and fills the set $\overline{\Delta}$.

In 1958, O.A.\,Ladyzhenskaya and L.D.\,Faddeev \cite{LaF1958} have
completely dropped the smallness requirement on $V$ and considered
the model operator
\begin{subequations}
    \label{FF}
\begin{align}
\tag{\ref{FF}\,\textrm{a}}\label{FF:a}
H&=H_0+V,  \\
\tag{\ref{FF}\,\textrm{b}}\label{FF:b}
(H_0 f)(\lambda)=\lambda f(\lambda),&\quad (Vf)(\lambda)=\int_\Delta V(\lambda,\mu)f(\mu) d\mu,\\
f\in & L_2(\Delta,\fh), \quad \Delta=(a,b), \nonumber
\end{align}
\end{subequations}
with NO small $\epsilon$ in front of $V$. Instead, they require compactness of the value
of $V(\lambda,\mu)$ as an operator in $\fh$ for any $\lambda,\mu\in\overline{\Delta}$.
\medskip

Proofs of the results in \cite{LaF1958} (and their extension) are
given in a Faddeev's 1964 work \cite{Fa1964}. In fact, this work
presents a complete version of the scattering theory for the
model \eqref{FF}. The 1964 paper may also be viewed as a
relatively simple introduction to the methods and ideas Faddeev used
in his celebrated analysis \cite{Fa1963} of the three-body problem.

Faddeev's detail study of the Hamiltonian \eqref{FF} is the first
reason why this Hamiltonian is sometimes called the Friedrichs-Faddeev
model. One more reason is related to the fact that the 1948
Friedrichs' paper contains another important ($2\times 2$ block
matrix) operator model that is called ``simply'' Friedrichs' model.
The second model works, in particular, for the Feshbach resonances (see,
e.g., \cite{GPr2011,AM2017} and references therein). For later
developments concerning the Friedrichs-Faddeev model proper and its
generalizations see \cite{PaPe1970,La1986c,La1989,DNY1992,IsRi2012}.

Notice that the typical two-body Schr\"odingrer operator may be
viewed as a particular case of the Friedrichs-Faddeev  model with
{$a=0$} and {$b=+\infty$} (see \cite{Fa1964}). Simply consider the c.m.
two-body Hamiltonian in the momentum $(\bk)$ space and make the variable change
{$|\bk|^2\to \lambda$}; in this case the internal (auxiliary) space
is $\fh=L_2(S^2)$, i.e. the space of square-integrable functions
on the unit sphere $S^2$ in $\bbR^3$.

It turned out that there is a certain gap in the study of analytical
properties and structure of the Friedrichs-Faddeev $T$- and
$S$-matrices on uphysical sheets of the energy plane. We fill this
gap by using the ideas and approach from the author's works
\cite{Mo1993}, \cite{Mo1997}. Namely, we derive representations for
the $T$- and $S$-matrices on uphysical sheets that explicitly
express them in terms of these same operators considered exclusively
on the physical sheet (see Proposition \ref{ProT} and Corollary
\ref{corol}). These representations show that the resonances
correspond, in fact, to the energies $z$ in the physical sheet where
the scattering matrix has eigenvalue zero.

Furthermore, we perform a {\textit{complex deformation}} (a
generalization of the complex scaling) of the Friedrichs-Faddeev
Hamiltonian. Discrete spectrum of the complexly deformed Hamiltonian
contains the {``complex scaling resonances''}. We show that these
resonances are simultaneously the {scattering matrix resonances}. In
the case of the Friedrichs-Faddeev model this is done quite easily and
illustratively.  Recall that, in general, to prove the equivalence
of scaling resonances and scattering matrix resonances is rather a
hard job (see \cite{Ha1979}).

Few words about notation used throughout the article. By $\sigma(T)$
we denote the spectrum of a closed linear operator $T$. Notations
$\sigma_p(T)$  and $\sigma_c(T)$ are used for the point spectrum
(the eigenvalue set) and continuous spectrum of $T$.  By
$I_{\mathfrak{K}}$ we denote the identity operator on a Hilbert (or
Banach) space $\mathfrak{K}$; the index $\mathfrak{K}$ may be
omitted if no confusion arises.

\section{Structure of the \textit{T}- and \textit{S}-matrices on unphysical energy sheets}
\label{Sec2}

First, let us recollect the description of the Friedrichs-Faddeev
model. We assume that {$\fh$} is an auxiliary
(``internal'') Hilbert space and {$\Delta=(a,b)$}, an interval on
{$\bbR$},
$$
{-\infty\leq a<b\leq+\infty}.
$$
Hilbert space of the problem is the space of square-integrable
{$\fh$}-valued functions on {$\Delta$}, $\fH=L_2(\Delta,\fh)$, with
the scalar product
$$
\lal f,g\ral=\int_a^b d\lambda \lal f(\lambda),g(\lambda)\ral_\fh,
$$
where {$\lal\cdot,\cdot\ral_\fh$} denotes the inner product in
$\fh$. Surely, the norm on {$\fH$} is given by
$$
\|f\|=\left(\int_a^b d\lambda \|f(\lambda)\|_\fh^2\right)^{1/2},
$$
where $\|\cdot\|_\fh$ stands for the norm on $\fh$.

Unperturbed Hamiltonian $H_0$ and perturbation potential $V$ are
given by \eqref{FF:b} where for each $\lambda,\mu\in(a,b)$ the value of
$V(\lambda,\mu)$ is a compact operator in $\fh$. We assume that
the function $V(\lambda,\mu)$ admits analytic continuation both in $\lambda$ and
$\mu$ into some domain {$\Omega\subset\bbC$} containing $\Delta$.
More precisely, we assume that
\begin{equation}
\label{Vhol}
V(\lambda,\mu) \quad \text{is holomorphic in both}\quad
\lambda,\mu\in\Omega,\quad (a,b)\subset\Omega.
\end{equation}
In addition, we suppose that $V(\lambda,\mu)=V(\mu,\lambda)^*$ for
real $\lambda,\mu \in\Delta$ (for Hermiticity of $V$). This
automatically implies $V(\lambda,\mu)=V(\mu^*,\lambda^*)^*$ for any
$\lambda,\mu\in\Omega$ such that their conjugates
$\lambda^*,\mu^*\in\Omega$ and, hence, the domain $\Omega$ should be
mirror symmetric with respect to the real axis. Following Friedrichs
\cite{Fr1948} and Faddeev \cite{Fa1964} we also require
\begin{equation}
\label{Vab}
V(a,\mu)=V(b,\mu)=V(\lambda,a)=V(\lambda,b)=0,\quad \text{ in case of finite $a$ or/and $b$}
\end{equation}
or impose suitable
requirements on the rate of decreasing of $V(\lambda,\mu)$ as
$|\lambda|,|\mu|\to\infty$, in case of infinite $a$ or/and $b$.
To simplify consideration, in the latter case we assume that
$\Omega$ and $V$ are such that
\begin{align}
\label{C1}
\|V(\lambda,\mu)\|&\leq K (1+|\lambda|+|\mu|)^{-(1+\eta_1)}, \quad \eta_1>0;\\
\label{C2}
\|V(\lambda+\alpha,\mu+\beta)-V(\lambda,\mu)\| &\leq K  (1+|\lambda|+|\mu|)^{-(1+\eta_1)}
(|\alpha|^{\eta_2}+|\beta|^{\eta_2}), \quad \eta_2>1/2,
\end{align}
for some $K>0$ and for any $\lambda,\mu\in\Omega$ and any $\alpha$,
$\beta$ such that $\lambda+\alpha\in\Omega$, $\mu+\beta\in\Omega$.
Since $V(\lambda,\mu)$ is holomorphic in $\Omega$ both in $\lambda$
and $\mu$, the requirement \eqref{C2} with $\eta_2<1$ is responsible
for the behavior of $V(\lambda,\mu)$ in the neighborhoods of
the (finite) end points $a$ and/or $b$. Otherwise, one can replace
$\eta_2$ with unity.
\medskip

As usually, the total Hamiltonian is $H=H_0+V$. Also we use the
standard notation for the resolvents and for the transition
operator: For $z$ lying outside the corresponding spectrum
$\sigma(H_0)$ or $\sigma(H)$, we introduce
$$
R_0(z)=(H_0-z)^{-1}, \quad R(z)=(H-z)^{-1},\quad\text{and}\quad T(z)=V-VR(z)V.
$$
Recall that, at least, for $z\not\in\sigma(H_0)\cup\sigma(H)$
\begin{equation}
\label{RTR}
R(z)=R_0(z)-R_0(z)T(z)R_0(z).
\end{equation}
Thus, the spectral problem for the Hamiltonian $H$ is reduced to the study of the
transition operator $T(z)$, the kernel of which is less singular
than that of $R(z)$.

From Faddeev's work \cite{Fa1964} we know that $T(\lambda,\mu,z)$ is
well-behaved function of $\lambda,\mu\in\Delta$ and $z$ on the
complex plane $\bbC$ punctured at $\sigma_p(H)$ and cut along
$(a,b)$. More precisely (see \cite[Theorem 3.1]{Fa1964}),
$T(\lambda,\mu,z)$ is of the same class \eqref{C1}, \eqref{C2} as
$V(\lambda,\mu)$ but with $\eta_1$ and $\eta_2$ replaced by positive
$\eta'_1<\eta_1$ and $\eta'_2<\eta_2$ which may be chosen arbitrary
close to $\eta_1$ and $\eta_2$, respectively. Furthermore, the
kernel $T(\lambda,\mu,z)$ has limits
$$
T(\lambda,\mu,E\pm i0), \quad E\in\Delta\setminus\sigma_p(H),
$$
that are (in our case) smooth in $\lambda,\mu\in\Delta\setminus\sigma_p(H)$. The
scattering matrix for the pair $(H_0,H)$ is given by
$$
S_+(E)=I_\fh-2\pi\ri \, T(E,E,E+\ri 0),\quad E\in(a,b)\setminus\sigma_p(H).
$$
Notice that due to the condition \eqref{Vab} and requirement
\eqref{C2} the eigenvalue set $\sigma_p(H)$ of $H$ consists of
finite number of eigenvalues having finite multiplicities (see
\cite{Fa1964}; cf. \cite{DNY1992}).

Now take a look of the Lippmann-Schwinger equations for the kernel
$T(\lambda,\mu,z)$ of the transition operator $T(z)$:
\begin{align}
\label{TLS1}
{T}({\lambda},\mu,z)&=V({\lambda},\mu)-\int_a^b d\nu\quad
\frac{V({\lambda},\nu){T}(\nu,\mu,z)}{\nu-z},
\\
\label{TLS2}
{T}(\lambda,{\mu},z)&=V(\lambda,{\mu})-\int_a^b d\nu\quad
\frac{{T}(\lambda,\nu,z)V(\nu,{\mu})}{\nu-z},
\\
\nonumber
&\qquad z\not\in(a,b),\quad \lambda,\mu\in(a,b)
\end{align}
Clearly, since $V(\lambda,\mu)$ is analytic in $\lambda\in\Omega$,
equality \eqref{TLS1} implies the same analyticity of $T(\lambda,\mu,z)$.
Analogously, equality  \eqref{TLS2} yields the holomorphy of {$T(\lambda,\mu,z)$}
in {$\mu\subset\Omega$}. Thus we arrive at the following
statement.

\begin{proposition}
\label{P1} If  $z\not\in(a,b)\cup\sigma_p(H)$, the kernel
$T(\lambda,\mu,z)$ is holomorphic in both $\lambda\in\Omega$ and
$\mu\in\Omega$. Furthermore, one can replace {$(a,b)$} in
\eqref{TLS1} and \eqref{TLS2} by arbitrary piecewise smooth Jordan
contour {$\gamma\subset\Omega$} obtained by continuous deformation
from {$(a,b)$} provided that the end points are fixed and the point
{$z$}  is avoided during the transformation $(a,b)\to\gamma$.
\end{proposition}

In the following {$\bbC^+=\{z\in\bbC\,|\,\, \Img z>0\}$} \quad
({$\bbC^-=\{z\in\bbC\,|\,\, \Img z<0\}$}) stands for the upper
(lower) halfplane of {$\bbC$}.

\begin{figure}[htb]
\includegraphics[angle=0.,width=5.cm]{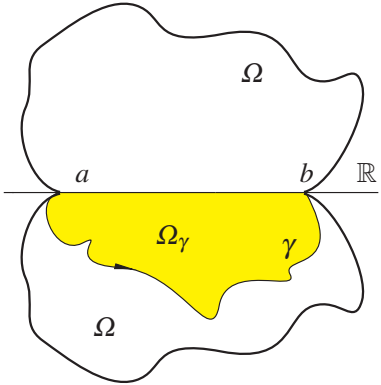}
\caption{Holomorphy domain $\Omega$ for the kernel $V(\lambda,\mu)$.
The set $\Omega_\gamma$ is bounded by (and contains both) the segment
$[a,b]$ and Jordan contour $\gamma$.
} \label{Fig1}
\end{figure}

To simplify the presentation, in the remainder of this note we
usually assume that the real numbers $a$ and $b$ are finite.

Suppose that $\gamma\subset\Omega\cap\bbC^\pm$ is a smooth Jordan
contour obtained from the interval $(a,b)$ by continuous
transformation with fixed end points $a$ and $b$. Then Proposition
\ref{P1} implies that one can equivalently rewrite \eqref{TLS1} as
\begin{align}
\label{TLS3}
T(\lambda,\mu,z)=&V(\lambda,\mu)-\int_\gamma d\nu\quad
\frac{V(\lambda,\nu)T(\nu,\mu,z)}{\nu-z},\\
\nonumber
& \lambda,\mu\in\Omega,\quad z\in\bbC\setminus{\Omega_\gamma},
\end{align}
where the set $\Omega_\gamma\subset\bbC$ is confined by (and
containing) the segment {$[a,b]$} and the curve {$\gamma$} (see
Figure \ref{Fig1}). One may almost literally repeat for \eqref{TLS3}
the analysis of the Lippmann-Schwinger equation \eqref{TLS1}
performed by Faddeev in \cite{Fa1964}. By applying to \eqref{TLS3}
the analytic Fredholm theorem \cite[Theorem VI.14]{RSI} one then
concludes that the solution $T(\lambda,\mu,z)$, in the appropriate
class of H\"older continuous kernels, exists (and is unique) except
for a discrete set of points that consists of the original point
spectrum $\sigma_p(H)$ of $H$ and an additional discrete set
$\sigma_{\rm res}(\gamma)$ located inside $\Omega_\gamma$. Moreover,
the solution $T(\lambda,\mu,z)$ to \eqref{TLS3} is analytic in $z$
for
\begin{equation}
\label{zgam}
z\not\in\sigma_p(H)\cup\overline{\gamma}\cup\sigma_{\rm
res}(\gamma),
\end{equation}
where the overlining in $\overline{\gamma}$ means the closure, that
is, $\overline{\gamma}=\gamma\cup\{a\}\cup\{b\}$. Again, because of
the holomorphy of $V(\lambda,\mu)$ in $\lambda,\mu\in\Omega$ the
solution $T(\lambda,\mu,z)$ remains analytic in
$\lambda,\mu\in\Omega$ for any $z\in\bbC$ satisfying \eqref{zgam}.
This is proven by the same reasoning as in the proof of the first
statement of Proposition \ref{P1}.    The points of $\sigma_{\rm
res}(\gamma)$ give to the solution $T(z)$ poles, residues at which
are finite rank operators. Thus, the equation \eqref{TLS3} gives us
an opportunity to pull the argument $z$ of $T(z)$ from the upper
half-plane $\bbC^+$ to the lower one, at least into the interior of
the set $\Omega_\gamma$. Surely, during this procedure one should
avoid the points of the discrete set $\sigma_{\rm res}(\gamma)$!

However, if, after such a pulling of $z$, one tries to re-establish
in \eqref{TLS3} the original integration over the interval $(a,b)$,
it is necessary to compute the residue at the pole $z$. That is, the
Lippmann-Schwinger equation \eqref{TLS3} changes its form and,
hence, for $z\in\Omega\cap\bbC^-$ the solution $T(\lambda,\mu,z)$ is
taken, in fact, on an unphysical sheet of the Riemann energy surface
of $T$. We denote this unphysical sheet by $\Pi_-$; it is attached
to the physical sheet via the upper rim of the cut along $(a,b)$.
Thus, we are forced to use a different notation, say
$T'(\lambda,\mu,z)$ for the continuation of the kernel of $T$ to
$\Pi_-$. However for $z$ outside $\Omega_\gamma$ this kernel
coincides with the original one, that is,
$T'(\lambda,\mu,z)=T(\lambda,\mu,z)$, provided
$z\in\bbC\setminus\bigl(\Omega_\gamma\cup\sigma_p(H)\bigr)$.

In fact we can solve the continued Lippmann-Schwinger equation
\eqref{TLS3} explicitly! To this end assume that $z\in\Omega_\gamma$
but $z\not\in\overline{\gamma}\cup\sigma_{\rm res}(\gamma)$, and
perform a reverse two-step transformation of the contour $\gamma$
(see Figure~\ref{Fig1}) back to the interval $\Delta=(a,b)$ in the
way shown in Figure~\ref{Fig2}.

\begin{figure}[htb]
{\includegraphics[angle=0,width=8cm]{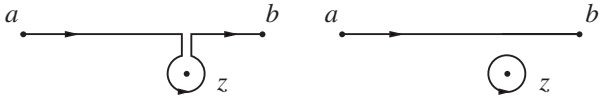}}
\caption{Two steps in deformation of the contour $\gamma$.}
\label{Fig2}
\end{figure}

After such a transformation and computing the residue at {$\nu=z$},
equation \eqref{TLS3} (written already for the unphysical-sheet values
$T'(\lambda,\mu,z)$) turns into the following equation:
\begin{align}
\label{TLS5}
& T'(\lambda,\mu,z)=V(\lambda,\mu)-2\pi\ri\,\, V(\lambda,z)T'(z,\mu,z)-\int_a^b d\nu\quad
\frac{V(\lambda,\nu)T'(\nu,\mu,z)}{\nu-z},\\
\nonumber
& \qquad\qquad\qquad\qquad \lambda,\mu\in\Omega, \quad \quad z\in\Omega\cap\bbC^-.
\end{align}
The points forming $\sigma_p(H)$ and $\sigma_{\rm res}(\gamma)$ for
various $\gamma$ will manifest themselves as singular points of the
equation \eqref{TLS5}. (We will discuss the independence of the
common part $\sigma_{\rm res}(\gamma_1)\cap\sigma_{\rm
res}(\gamma_2)$ of the sets $\sigma_{\rm res}(\gamma_1)$ and $\sigma_{\rm
res}(\gamma_2)$ for different $\gamma_1\neq\gamma_2$ later, in Section \ref{Sec3}.)

Following the standard terminology of scattering theory in momentum
space we call the kernel $T'(z,\mu,z)$ ``half-on-shell'' since its
first argument equals the spectral parameter (energy) $z$. The
kernel $T'(z,z,z)$ is called ``on-shell'' since both the first and
second arguments are equal to the energy $z$. Finally, the kernel
$T'(\lambda,\mu,z)$ with arbitrarily chosen admissible values of the
arguments $\lambda$ and $\mu$ is called ``off-shell''. The terms
``off-shell'', ``half-on-shell'', and ``on-shell'' may be applied to
any function of the three (complex) arguments $\lambda$, $\mu$, and
$z$.

Let us transfer in \eqref{TLS5} all the entries with the off-shell $T'$ to the left-hand
side and then obtain
\begin{align}
\label{TLS6}
&T'(\lambda,\mu,z)+\int_a^b d\nu\quad
\frac{V(\lambda,\nu)T'(\nu,\mu,z)}{\nu-z}=
V(\lambda,\mu)-2\pi\ri\, V(\lambda,z)T'(z,\mu,z),\\
\nonumber
& \qquad\qquad\qquad\qquad \lambda,\mu\in\Omega, \quad z\in\Omega\cap\bbC^-.
\end{align}
Equation \eqref{TLS6} allows us to express the off-shell $T'$ exclusively through
the half-on-shell $T'$ by taking into account that, on the physical sheet,
$$
{(I+VR_0(z))}T(z)=V \quad \Longrightarrow \quad {(I+VR_0(z))^{-1}}V=T(z),\qquad z\not\in\sigma_p(H).
$$
Thus, \eqref{TLS6} implies
\begin{align}
\label{TLS7}
& T'(\lambda,\mu,{z})=
T(\lambda,\mu,z)-2\pi\ri\, T(\lambda,{z},{z})T'({z},\mu,{z}),
\end{align}
where the notation $T(\cdot,\mu,z)$ with no prime means that this
entry is taken for $z$ on the physical sheet of the energy plane.

By setting $\lambda=z$ in \eqref{TLS7} we get ${T'({z},\mu,{z})=
T({z},\mu,z)-2\pi\ri\, T({z},{z},{z})T'({z},\mu,{z})}$, which
yields
\begin{equation}
\label{STT}
{S_-(z)T'({z},\mu,{z})=T({z},\mu,z)},
\end{equation}
where
\begin{equation}
\label{Sz}
S_-({z}):=
I_\fh+2\pi\ri\, T({z},{z},{z}),\quad z\in\Omega\cap\bbC^-,
\end{equation}
is nothing but the scattering matrix for $z$ in the lower complex half-plane. We
underline that the argument $z$ of $S_-(z)$ in \eqref{Sz},
$z\in\Omega\cap\bbC^-$, lies on the physical sheet. From \eqref{STT}
it follows that
\begin{equation}
\label{TST}
T'({z},\mu,{z})=S_-(z)^{-1}T(z,\mu,z),
\end{equation}
provided $z$ is not a point where $S_-(z)$ has an eigenvalue zero.
Finally, the relations \eqref{TLS7} and \eqref{TST} yield
\begin{align}
\label{TLS8} &
T'(\lambda,\mu,{z})= T(\lambda,\mu,z)-2\pi\ri\,
T(\lambda,{z},{z}){S_-(z)^{-1}}T({z},\mu,{z}).
\end{align}
All the entries on the r.h.s. part of \eqref{TLS8} are taken on the physical sheet.

In a similar way we perform the continuation of $T(\lambda,\mu,z)$
from the lower half-plane $\bbC^-$ to the part $\Omega\cap\bbC^+$ of
the unphysical energy sheet {$\Pi_+$} attached to the physical sheet
along the lower rim of the cut $(a,b)$. As a result, we arrive with
the following combined statement (for both the sheets $\Pi_\ell$
where the symbol $\ell=\pm1$ is identified with the respective sign $\pm$ in
the previous notation $\Pi_\pm$).

\begin{proposition}
\label{ProT} The transition operator $T(z)$ admits meromorphic
continuation (as a $\cB(\fH)$-valued function of the variable $z$) through the cut along $(a,b)$ from both the upper,
$\bbC^+$, and lower, $\bbC^-$, half-planes to the respective parts
$$
\text{$\Omega_{-1}:=\bbC^-\cap\Omega$ and $\Omega_{+1}:=\bbC^+\cap\Omega$}
$$
of the unphysical sheets
$\Pi_{-1}$ and $\Pi_{+1}$ adjoining the physical sheet along the
upper and lower rims of the mentioned cut. The kernel of the
continued operator $T(z)\bigr|_{\Pi_\ell\cap\Omega_\ell}$, $\ell=\pm1$,
admits the representation
\begin{align}
\label{TPl}
T(\lambda,\mu,{z})\bigr|_\text{{$z\in\Pi_\ell\cap\Omega_\ell$}}&=
\bigl(T(\lambda,\mu,z)+2\pi{\rm i}\,\ell\,\,
T(\lambda,{z},{z}){S_\ell(z)^{-1}}
T({z},\mu,{z})\bigr)\bigr|_\text{$z\in{\Omega_\ell}$},\\
&\qquad z\in\Omega_\ell\setminus\sigma^\ell_{\rm res},
\end{align}
where all the terms on the r.h.s. part, including the scattering matrix
\begin{equation}
\label{Spm}
S_\ell({z})=
I_\fh- 2\pi\mathrm{i}\ell\, T({z},{z},{z}),
\end{equation}
are taken for $z$ on the physical sheet, and $\sigma^\ell_{\rm res}$ denotes the set of all those points
$\zeta\in\Omega\cap\bbC^\ell$ for which $S_\ell({\zeta})$ has eigenvalue zero.
\end{proposition}
\begin{remark}
Whether {$\Pi_-$} and {$\Pi_+$} represent the same (``second'')
unphysical sheet, depends on the analytical properties of
{$V(\lambda,\mu)$} outside {$\Omega$} (if available, cf.
\cite{Mo1993}).
\end{remark}

Continuation formula \eqref{TPl} for $T(z)$ implies the following important consequence.
\begin{corollary}
\label{corol}
Analytic continuation of the scattering matrix
$S_{-\ell}(z)$, $\ell=\pm1$, to the unphysical sheet $\Pi_{\ell}$ is described by the equality
\begin{equation}
\label{SlS}
S_{-\ell}({z})\bigr|_{z\in\Pi_{\ell}\,\cap\Omega_\ell}=
S_{\ell}(z)^{-1}\bigr|_{z\in\Omega_\ell},\qquad z\not\in\sigma^{\ell}_{\rm res},
\end{equation}
where the r.h.s. part is taken for $z$ on the physical sheet.
\end{corollary}

\section{Complex scaling and Friedrichs-Faddeev model}
\label{Sec3}

Let $\Delta_{\br}$ denote the Laplacian in the variable
${\br}\in\bbR^3$. In the coordinate space, the standard complex
scaling \cite{Lo1964}, \cite{BaC1971} means the replacement of the
original c.m. two-body Hamiltonian
\begin{equation}
\label{HS}
H=-\Delta_{\br} + \widehat{V}(\br)
\end{equation}
by the non-Hermitian operator
\begin{equation}
\label{Htet}
H(\theta)=-{\rm e}^{-2\ri \theta}\Delta_{\br}+\widehat{V}({\rm e}^{\ri \theta}\br),
\end{equation}
for a non-negative $\theta\leq\pi/2$, provided that the local
potential $\widehat{V}(\br)$ admits  analytic continuation to a
domain of complex $\bbC^3$-arguments $\br$. Location of the spectrum
of $H(\theta)$ is shown schematically in Fig. \ref{Fig3}.

\begin{figure}[htb]
{\includegraphics[angle=0,width=9.cm]{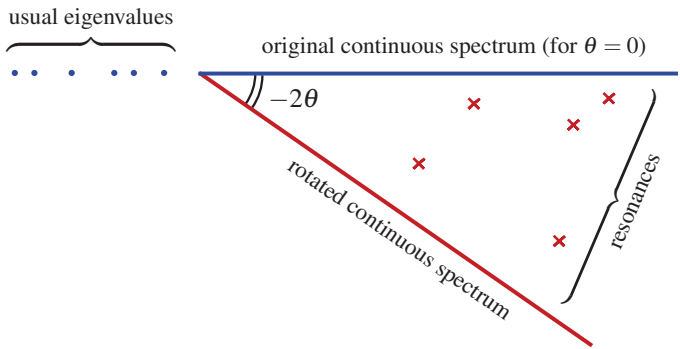}}
\caption{Spectrum of the complexly ``rotated'' Hamiltonian
$H(\theta)$.} \label{Fig3}
\end{figure}

Having performed the Fourier transform of  \eqref{Htet} and then
making the change $|\bk|^2\to\lambda$ one arrives at the complex
version of the Friedrichs-Faddeev model
\begin{align}
\label{FF1}
(H(\theta)f)(\lambda)& ={\rm e}^{-2\ri \theta}\lambda f(\lambda)+
{\rm e}^{-2\ri \theta}\int_0^\infty {V}({\rm e}^{-2\ri \theta}\lambda,
{\rm e}^{-2\ri \theta}\mu)f(\mu)d\mu,\\
\nonumber
& f\in L_2\bigl(\bbR^+,L_2(S^2)\bigr).
\end{align}
The operator-valued function ${V}(\lambda,\mu)$ is explicitly
expressed  through the Fourier transform of $\widehat{V}$. For every
admissible $\lambda,\mu\in\bbC$ the value of ${V}(\lambda,\mu)$ is
an operator (typically, compact) in $\fh=L_2(S^2)$.

The Hamiltonian \eqref{FF1} may be interpreted as an analogue of the
Fri\-ed\-richs-Fad\-de\-ev model \eqref{FF} for a contour in the complex plane,
\begin{equation}
\label{Hgf}
(H_\gamma f)(\lambda)=\lambda f(\lambda)+
\int_\gamma {V}(\lambda,\mu)f(\mu)d\mu,\qquad \lambda\in\gamma,
\end{equation}
where
$$
\gamma={\rm e}^{-2\ri \theta}\bbR^+:=\{z\in\bbC\,|\,\,z={\rm e}^{-2\ri \theta}x,\,0\leq x<\infty\}
$$
and $f\in L_2\bigl(\gamma,L_2(S^2)\bigr)$.

Surely, in the two-body problem case, one has to assume that
${V}(\lambda,\mu)$ is analytic in both $\lambda$ and $\mu$ on some
domain $\Omega\subset\bbC$ containing the positive semiaxis
$\bbR^+$. In addition, $\|{V}(\lambda,\mu)\|$ should decrease
sufficiently rapidly as $|\lambda|\to\infty$ and/or $|\mu|\to\infty$
(in order to ensure that the integral on the right-hand side of \eqref{Hgf}
defines a reasonable operator in $L_2\bigl(\gamma,L_2(S^2)\bigr)$).

\section{Equivalence of the complex rotation resonances and
scattering resonances \newline in the Friedrichs-Faddeev  model}
\label{Sec4}

In the following we consider a family of the
Friedrichs-Faddeev  Hamiltonians
\vspace*{-0.1cm}
\begin{align*}
{H_\gamma}&{=H_{0,\gamma}+V_\gamma}
\end{align*}
associated with smooth Jordan curves $\gamma \subset\Omega$ originating in
$(a,b)$. As before, notation $\Omega$ is used for the holomorphy
domain of $V(\lambda,\mu)$ in the variables $\lambda$ and $\mu$;
$\Omega$ may or may not include $a$ and/or $b$;
\begin{align*}
{\bigl(H_{0,\gamma}f\bigr)(\lambda)}&{=\lambda\,f(\lambda)}\quad
\text{ and }\quad
{\bigl(V_\gamma f\bigr)(\lambda)=\int_\gamma
{V}(\lambda,\mu)f(\mu)d\mu,\quad \lambda\in\gamma}.
\end{align*}
Here, $f$ is taken from the the Hilbert space $\fH_\gamma=L_2(\gamma ,\fh)$
by which one understands the space of $\fh$-valued functions of the
variable $\lambda\in\gamma$ with the scalar product
$$
\lal f,g\ral_\gamma=\int_\gamma  |d\lambda| \lal f(\lambda),g(\lambda)\ral_\fh.
$$

Again assume that both $a$ and $b$ are finite and let
$V(\lambda,\mu)$ be as in Section \ref{Sec2}. Introduce the
transition operator for the pair $(H_{0,\gamma},H_\gamma)$:
\begin{equation}
\label{Tgz}
T_\gamma(z)=V_\gamma-V_\gamma (H_\gamma-z)^{-1}V_\gamma, \quad z\not\in\sigma(H_\gamma).
\end{equation}
For $R_\gamma(z)=(H_\gamma-z)^{-1}$ we have
\begin{equation}
\label{RgT0}
{R_\gamma(z)=R_{0,\gamma}(z)-R_{0,\gamma}(z){T_\gamma(z)}R_{0,\gamma}(z)},
\end{equation}
with $R_{0,\gamma}(z)=(H_{0,\gamma}-z)^{-1}$, $z\not\in\sigma(H_{0,\gamma})$.

Notice that $H_{0,\gamma}$  is a normal operator on $\fH_\gamma$ and
it has only absolutely continuous spectrum that coincides with the
curve $\overline{\gamma}$. Thus, from \eqref{RgT0} it follows that  the
discrete eigenvalues of $H_\gamma$ are nothing but the poles of the
operator-valued function $T_\gamma(z)$.

Assume that the above Jordan curve $\gamma$ lies completely in
$\Omega_-=\Omega\cap\bbC^-$ (or completely in
$\Omega_+=\Omega\cap\bbC^+$).  Let  again $\Omega_\gamma$ denote the
set in the complex plane $\bbC$ confined by (and containing) the
interval $[a,b]$ and the contour $\gamma$.

\begin{proposition}
\label{Pfin} The part of the spectrum of $H_\gamma$ lying outside
$\Omega_\gamma$ is purely real and coincides with
$\sigma_p(H)\setminus\overline\Delta$. Furthermore,
$\sigma_p(H_\gamma)\cap\Delta=\sigma_p(H)\cap\Delta$, that is, the
point spectrum eigenvalues of $H_\gamma$ lying on $\Delta$ do
not depend on the (smooth) Jordan contour $\gamma$. The spectrum of
{$H_\gamma$} inside {$\Omega_\gamma$} represents the
scattering-matrix resonances.
\end{proposition}
\begin{proof}
Already from \eqref{Tgz} one may conclude that, for any fixed
$z\not\in\sigma(H_\gamma)$, the kernel $T_\gamma(\lambda,\mu,z)$ is
holomorphic in the variables $\lambda,\mu\in\Omega$ (since $V$ is
holomorphic). Indeed, \eqref{Tgz} means
$$
{T_\gamma({\lambda,\mu},z)=V({\lambda,\mu})+\int_{\gamma}d\mu'\int_\gamma d\lambda'\,\,
V({\lambda},\mu')R_\gamma(\mu',\lambda',z)V(\lambda',{\mu})}.
$$
One may pull $\lambda$ and $\mu$ anywhere in $\Omega$. And this will be true after analytic continuation
of $R_\gamma(\mu',\lambda',z)$ in $z$ through $\gamma$.

Now look at the Lippmann-Schwinger equation for $T_\gamma$,
\begin{equation}
\label{TLSg}
T_\gamma(\lambda,\mu,z)=V(\lambda,\mu)-\int_\gamma d\nu\quad\frac{V(\lambda,\nu)T_\gamma(\nu,\mu,z)}{\nu-z},
\quad z\not\in\overline{\gamma},\quad \lambda,\mu\in\gamma.
\end{equation}
Assume that $z\in\bbC\setminus\Omega_\gamma$ and consider for such a $z$ also the Lippmann-Shwinger
equation for the ``original'' $T$-matrix --- it is associated with
the interval $\Delta$:
\begin{equation}
\label{TLS}
T(\lambda,\mu,z)=V(\lambda,\mu)-\int_a^b d\nu\quad\frac{V(\lambda,\nu)T(\nu,\mu,z)}{\nu-z},
\quad z\not\in\Omega_\gamma,\quad \lambda,\mu\in(a,b).
\end{equation}
Since both the kernels $V(\lambda,\cdot,z)$ and $T(\lambda,\cdot,z)$
for fixed $z\not\in \Omega_\gamma (\cup\sigma_p(H))$ are holomorphic
in $\lambda\in\Omega$, one may transform the interval $[a,b]$ into
the contour $\gamma$ and obtain:
\begin{equation}
\label{TLSgn}
T(\lambda,\mu,z)=V(\lambda,\mu)-\int_\gamma d\nu\quad\frac{V(\lambda,\nu)T(\nu,\mu,z)}{\nu-z},
\quad z\not\in\Omega_\gamma,\quad \lambda,\mu\in(a,b).
\end{equation}
Pull $\lambda,\mu$ on $\gamma$ and then compare \eqref{TLSg} and \eqref{TLSgn}. The uniqueness theorem for the solution to
\eqref{TLSgn} implies:
$$
\text{{{$T_\gamma(\lambda,\mu,z)=T(\lambda,\mu,z) \quad \text{whenever}
\quad \lambda,\mu\in\gamma,\,\, z\in\Omega\setminus\Omega_\gamma$}}}
\text{ (and } z\not\in \sigma_p(H)).
$$

\begin{figure}[htb]
{\includegraphics[angle=0,width=7.cm]{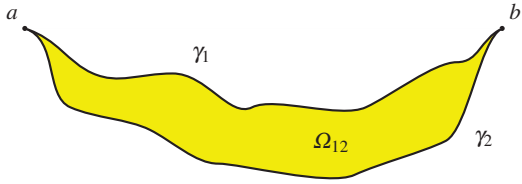}}
\caption{The closed set $\Omega_{12}$ bounded by the Jordan contours $\gamma_1$ and $\gamma_2$.}
\label{Fig4}
\end{figure}

\noindent Similarly, we have
\begin{align}
T_{\gamma_1}(\lambda,\mu,z)&=T_{\gamma_2}(\lambda,\mu,z)
\quad \text{whenever } \lambda,\mu\in\Omega  \text{ and }
z\not\in\Omega_{12}\cup\sigma_p(H_{\gamma_1}),
\nonumber
\end{align}
where $\Omega_{12}$
is the closed set bounded by the curves $\gamma_1$ and $\gamma_2$ (see Figure \ref{Fig4}).
This also means that
$$
\sigma_p(H_{\gamma_1})\setminus\Omega_{12}=\sigma_p(H_{\gamma_2})\setminus\Omega_{12}.
$$
Finally, by the uniqueness principle for analytic continuation, for
{$z$} inside {$\Omega_\gamma$} the kernel
{$T_\gamma(\lambda,\mu,z)$} represents just the analytic
continuation of {$T(\lambda,\mu,\cdot)$} to the interior of
{$\Omega_\gamma$} lying in the unphysical sheet. Hence, the poles of
{$T_\gamma(z)$} within {$\Omega_\gamma$} represent resonances of the
original Friedrichs-Faddeev Hamiltonian (the one that was introduced
for the interval $(a,b)$). This also means that the positions of
the resonances inside $\Omega_\gamma$ do not depend on {$\gamma$}.
The proof is complete.
\end{proof}

\section*{Conclusion}

For the (analytic) Friedrichs-Faddeev model, we have derived
representations that explicitly express the transition operator and
scattering matrix on unphysical energy sheets in terms of these same
operators considered exclusively on the physical sheet. A resonance
on a sheet {$\Pi_\ell$}, $\ell=\pm1$, or, more precisely, in the domain
$\Pi_\ell\cap\Omega_\ell$ is just a point, for the copy $z$ of which on
the physical sheet the corresponding scattering matrix
{$S_{\ell}(z)$} has eigenvalue zero, that is,
\begin{equation}
\label{Slz}
{S_\ell(z)\mathcal{A}=0}\quad\text{for a non-zero vector \,}{\mathcal{A}\in\fh}.
\end{equation}
Furthermore, we have shown that, for the Friedrichs-Faddeev model,
the deformation resonances are exactly the scattering matrix
resonances.

\end{document}